\documentclass[a4paper,UKenglish,cleveref, autoref]{lipics-v2019}
\nolinenumbers

\usepackage{amsmath,amssymb,array}
\usepackage{verbatim}
\usepackage{float}
\usepackage{tikz}
\usetikzlibrary{arrows,snakes,matrix,backgrounds}

\usepackage{microtype}

\usepackage[algoruled,linesnumbered,noend]{algorithm2e}

\usepackage{comment}
\usepackage{xcolor}

\mathchardef\mhyphen="2D

\newcommand{\DenseSub}{{\sf Densest Subgraph}}

\newcommand{\kDenseSub}{{\sf k-Densest Subgraphs}}

\newtheorem{problem}{Problem}
\title{Computing the $k$ Densest Subgraphs of a Graph}

\author{Riccardo  Dondi}{Universit\`a degli Studi di Bergamo, Bergamo, Italy }{riccardo.dondi@unibg.it}{}{}

\author{Danny Hermelin}{Ben-Gurion University of the Negev, Be'er Sheva, Israel}{hermelin@bgu.ac.il}{}{}

\authorrunning{ } 

\Copyright{ } 



\begin{document}

\maketitle

\begin{abstract}
Computing cohesive subgraphs is a central problem in graph theory. While many formulations of cohesive subgraphs lead to NP-hard problems, finding a densest subgraph can be done in polynomial-time. As such, the densest subgraph model has emerged as the most popular notion of cohesiveness. Recently, the data mining community has started looking into the problem of computing $k$ densest subgraphs in a given graph, rather than one.
In this paper we consider a 
natural variant of the $k$ densest subgraphs problem, 
where overlap between solution subgraphs is allowed with no constraint. We show that the problem is fixed-parameter tractable with respect to $k$, and admits a PTAS for constant $k$. Both these algorithms complement nicely the previously known $O(n^k)$ algorithm for the problem. 
\end{abstract}



\section{Introduction}
\label{sec:Introduction}
Finding cohesive subgraphs is a central problem in the analysis of social networks~\cite{DBLP:journals/cn/KumarRRT99}, graph-mining~\cite{DBLP:conf/kdd/SozioG10,DBLP:conf/www/TattiG15,DBLP:conf/www/Tsourakakis15a,DBLP:journals/tkdd/Tatti19},
group dynamics research~\cite{DyaramKamalanabhan}, computational biology~\cite{bioinfo2006}, and many other areas. The most basic and natural attempt at modeling cohesiveness is via the notion of cliques; however, this notion is too strict and rigid for most applications, and is also known to be computationally hard~\cite{DBLP:conf/coco/Karp72,DBLP:journals/toc/Zuckerman07}.

While there are several alternative definitions for cohesiveness~\cite{Komusiewicz16}, a notion that has emerged as arguably the most popular is the {\em densest subgraph} model~\cite{DBLP:conf/waw/AndersenC09,DBLP:conf/wsdm/BalalauBCGS15,DBLP:journals/datamine/GalbrunGT16,DBLP:conf/cikm/NasirGMG17,DBLP:conf/kdd/SozioG10,DBLP:conf/www/TattiG15,Zou2013}. Here, the \emph{density} of a graph is simply the edge-to-vertex ratio in the graph, and the densest subgraph is the (induced) subgraph that maximizes this ratio. As opposed to the maximum clique, finding a densest subgraph in a graph is polynomial-time solvable~\cite{DBLP:journals/siamcomp/GalloGT89,Goldberg:1984:FMD:894477,DBLP:journals/algorithmica/KawaseM18,PicardQ82}. This fact, along with the naturality of the concept, has lead the notion of density to nowadays be considered at the core of large
scale data mining~\cite{DBLP:journals/pvldb/BahmaniKV12}.

Recent contributions have shifted the interest from
computing a single cohesive subgraph to computing 
a set of such subgraphs~\cite{DBLP:conf/wsdm/BalalauBCGS15,DBLP:journals/datamine/GalbrunGT16,DBLP:conf/cikm/NasirGMG17,ValariKP12}, as this is naturally more desirable in most applications. The proposed approaches may allow (but not force) the subgraphs to overlap, as many real-world cohesive groups share common elements. For example, hubs may belong to more than one community~\cite{DBLP:journals/im/LeskovecLDM09,DBLP:journals/datamine/GalbrunGT16}. The way the overlap is restricted, if at all, varies among the different approaches. For instance, in~\cite{DBLP:conf/wsdm/BalalauBCGS15}, the notion of overlap is restricted via a constraint on the pairwise Jaccard coefficient of the subgraphs of the solution, while in~\cite{DBLP:journals/datamine/GalbrunGT16} the total overlap is factored into the objective function.

\subsection{A natural variant}

In this paper we consider a variant of the problem of computing $k$ densest subgraphs of a given graph,
where  subgraphs in the solution must be distinct (\emph{i.e.} have different vertex sets). Thus a solution subgraph may be a subgraph, a supergraph, or have almost the same vertex set as another solution subgraph. 
The objective function is the maximization
of the total sum of densities of the solution subgraphs. 

\begin{problem}{\kDenseSub} \\
\label{prob:distinct}%
\noindent
\emph{Input:} A graph $G$. \\
\emph{Output:} A set of $k$ pairwise distinct subgraphs $G_1,\ldots,G_k$ of $G$.\\
\emph{Objective:} Maximize $\sum_{i=1}^k density(G_i)$.
\end{problem}
While \kDenseSub{} is arguably the most basic variant for the problem of computing the $k$ densest subgraphs of a given graph, very little is known about the problem from a theoretical perspective. In~\cite{TOPK2019}, it is shown that this problem is solvable in $n^{O(k)}$ time. This is the main yardstick by which we assess the results in this paper.
\begin{theorem}[\cite{TOPK2019}]
\label{thm:previous}%
\kDenseSub{} can be solved in $n^{O(k)}$ time.
\end{theorem}



Our first result shows that there is a a rather efficient algorithm (for constant values of $k$), if one is willing to slightly compromise the quality of the solution. In particular, we show that the problem admits an efficient PTAS (EPTAS): 
\begin{theorem}
\label{thm:distinct}%
For any fixed $k, \varepsilon \geq 1$, there is an algorithm that computes in $O(m n \log n)$ time a $(1-\frac{1}{\varepsilon})$-approximate solution for \kDenseSub{}. \end{theorem}

Our second result shows that \kDenseSub{} is in fact fixed-parameter tractable when parameterized by the number $k$ of subgraphs. In particular, our second algorithm shows that the problem is polynomial-time solvable even for $k=\Theta(\lg n)$. More precisely, we prove the following:
\begin{theorem}
\label{thm:distinctFPT}%
\kDenseSub{} can be solved in $O(2^k m n^3 \log n)$ time. 
\end{theorem}


\subsection{Related work}

The \DenseSub{} problem, the problem of computing a densest subgraph in a given graph, is the special case of 
\kDenseSub{} when $k=1$. This problem has been extensively studied in the literature, and we outline here only the main results. The problem is known to be polynomial-time solvable~\cite{Goldberg:1984:FMD:894477,PicardQ82,DBLP:journals/siamcomp/GalloGT89,DBLP:journals/algorithmica/KawaseM18}, and it can be approximated within a factor of $\frac{1}{2}$
in linear time~\cite{DBLP:journals/jal/KortsarzP94,DBLP:conf/swat/AsahiroITT96,DBLP:conf/approx/Charikar00}. Generalization of the problem to weighted graphs~\cite{Goldberg:1984:FMD:894477}, as well as directed graphs~\cite{DBLP:conf/icalp/KhullerS09}, also turn out to be polynomial-time solvable. However, the \DenseSub{} problem becomes NP-hard when constraints on the number of vertices in the output graph are added~\cite{DBLP:conf/waw/AndersenC09,DBLP:journals/dam/AsahiroHI02,DBLP:journals/algorithmica/FeigePK01,DBLP:journals/corr/abs-0912-5327,DBLP:conf/icalp/KhullerS09,DBLP:conf/stoc/Manurangsi17}. 




\section{Preliminaries}
\label{sec:Pre}%

All graphs considered in this paper are simple, undirected, and without self-loops. Throughout the paper we let $G=(V,E)$ denote an input graph, and we let $n=|V|$ and $m=|E|$. For a vertex $v \in V$, we let $deg(v)$ denote the \emph{degree} of $v$ in $G$, \emph{i.e.} $deg(v) = |\{u \in V: \{u,v\} \in E\}|$. The density of $G$ is defined by  $density(G)=m/n$, and in general, the density of a graph is the ratio between the number of edges and the number of vertices in the graph. 

Given a subset of vertices $V_1 \subseteq V$, we denote by $G[V_1]$ the \emph{subgraph} of $G$ induced by $V_1$; formally, $G[V_1]=(V_1,E_1)$ where $E_1= \{\{ u,v \} \in E: u,v \in V_1\}$. Thus, a subgraph of $G$ is determined completely by its subset of vertices. If $G[V_1]$ and $G[V_2]$ are both subgraphs of $G$, then we say that these subgraphs are \emph{distinct} whenever $V_1 \neq V_2$. If $V_1 \cap V_2 = \emptyset$ then the two subgraphs are \emph{disjoint}, and if $V_1 \subset V_2$, then $G[V_2]$ is a \emph{proper supergraph} of $G[V_1]$.


\subsection{Goldberg's algorithm}
\label{sec:Goldberg}%

As mentioned above, the \DenseSub{} problem can be solved in polynomial-time~\cite{Goldberg:1984:FMD:894477,PicardQ82,DBLP:journals/siamcomp/GalloGT89}. The main idea is to reduce the problem to a series of min-cut computations. Picard and Queyranne's algorithm~\cite{PicardQ82} requires $O(n)$ such computations, where $n$ is the number of vertices in the input graph, while Goldberg's algorithm~\cite{Goldberg:1984:FMD:894477} improves this to $O(\log n)$, thus
giving an overall time complexity of $O(m n \log n)$ via Orlin's algorithm~\cite{Orlin13}. 
 Recently, the time complexity of Goldberg's algorithm for 
 unweighted graphs has been improved to $O(n^3)$ \cite{DBLP:journals/algorithmica/KawaseM18}.
Goldberg also showed that one can compute in $O(m n \log n)$ time a densest subgraph in a \emph{vertex-weighted} graph; here, the density of a vertex-weight graph $H$ on $n$ vertices of total weight $w$ and $m$ edges is given by $density(H)=(m+w)/n$.  


\section{An EPTAS for \kDenseSub{}}
\label{sec:PTAS}

In the following section we describe our EPTAS for \kDenseSub{}. Let $(G,k)$ denote a given instance of \kDenseSub{}, and let $\varepsilon > 0$ be a given constant. Our goal is to compute in $O(m n \log n)$ time $k$ distinct subgraphs $G_1,\ldots,G_k$ of $G$ with densities $d_1,\ldots,d_k$ such that $\sum_i d_i \geq (1-\frac{1}{\varepsilon}) \cdot OPT$, where $OPT$ is the
value of an solution of \kDenseSub{},
that is the
total sum of densities of the $k$ densest subgraphs in $G$. 
Recall that $k=O(1)$.

Below we first provide a description of our algorithm, followed by an analysis of its running time, and an analysis of its approximation ratio guarantee. Since the function $(\frac{n-2k}{n})^k$ tends to 1 as $n$ grows to infinity, we will henceforth assume that $n$ is sufficiently large so that the following inequality holds
(otherwise we can solve the problem optimally via brute force in $O(1)$ time):

\vspace{-.6cm}
\begin{equation}
\label{eqn:sufficientlylarge}
\left(\frac{n-2k}{n}\right)^k \geq \left(1-\frac{1}{\varepsilon}\right).
\end{equation}

\vspace{-.4cm}
\subsection{The algorithm}

We say that a subgraph $G_i=(V_i,E_i)$ of $G$ is \emph{small} if  $|V_i| \leq \varepsilon -1$. Our algorithm proceeds in a certain way so long that all subgraphs computed so far are small; once a subgraph which is not small is computed, the algorithm proceeds in a different manner. The first subgraph $G_1=(V_1,E_1)$ is computed using Goldberg's algorithm, so $G_1$ is a densest subgraph in $G$.

Suppose that we have computed subgraphs $G_1,\ldots,G_i$ for some $1 \leq i \leq k-1$, and all these subgraphs are small. The subgraph $G_{i+1}$ is taken to be a densest graph out of all of the following possible candidates:
\begin{itemize}
    \item A densest subgraph in $G[V \setminus \{v_1,\ldots,v_i\}]$ for some $v_1 \in V_1,\ldots,v_i \in V_i$.
    \item A densest strict supergraph of $G_j$ in $G$ for some $j \in \{1,\ldots,i\}$. 
\end{itemize}
Note that some of the candidates of the second type above can be graphs in $\{G_1,\ldots,G_i\}$; such graphs are naturally excluded from being candidates for the subgraph $G_{i+1}$.

Suppose that we have computed subgraphs $G_1,\ldots,G_i$ for some $1 \leq i \leq k-1$, and $G_i=(V_i,E_i)$ is not small. Then in this case $G_i$ can either be big or huge. We say that $G_i$ is \emph{big} if $|V_i| \leq n-k-i$, and otherwise it is \emph{huge}. If $G_i$ is big, we choose  arbitrary distinct vertices $v_{i+1},\ldots,v_k \in V \setminus V_i$ and set $G_j$ to be the graph induced by $V_i \cup \{v_j\}$ for $j\in\{i+1,\ldots,k\}$. Note that since $V_i$ is not huge, there are enough distinct vertices in $V \setminus V_i$. Also note that as $G_i$ is the only big subgraph in $G_1,\ldots,G_i$, it is not a proper subgraph of any of these graphs and so all subgraphs $G_j$ are distinct from all subgraphs computed so far. 

If $G_i$ is huge, then the graphs $G_{i+1},\ldots,G_k$ are computed by iteratively removing minimal degree vertices in $G_i$. Since $G_i$ is huge and all graphs $G_1,\ldots,G_{i-1}$ are small, we are guaranteed that subgraphs computed in this way are distinct from those we have computed. 
\subsection{Run-time analysis}

Before analyzing the run-time of our algorithm, we begin with the following lemma: 

\begin{lemma}
\label{lem:supergraph}%
Let $H_0$ be a strict subgraph of $G$, and let $H$ be a densest strict supergraph of $H_0$ in $G$. If $density(H) \leq density(H_0)$, then there is an algorithm that computes in $O(mn \log n)$ time a strict supergraph of $H_0$ in $G$ with density equal to $density(H)$, given $H_0$ as input.
\end{lemma}

\begin{proof}
Given $H_0=(V_0,E_0)$ as input, the algorithm uses Goldberg's algorithm to compute a densest subgraph $H_1=(V_1,E_1)$ in the vertex-weighted graph $G^*=G[V\setminus V_0]$, with vertex weights defined by $w(v) = |N_G(v) \cap V_0|$ for each vertex $v$ of $G^*$. It then returns the graph $H = H_0 \cup H_1 = G[V_0 \cup V_1]$ as a solution. Clearly, this can be done in $O(mn \log n)$ time, and $H$ is a strict supergraph of $H_0$ in $G$. We claim that $H$ is indeed a densest among all supergraphs of $H_0$. 

Let $H'=(V',E')$ be any strict supergraph of $H_0$ ($V_0 \subset V'$), and let $H_2=(V_2,E_2)$ be the subgraph of $G$ induced by $V_2=V' \setminus V_0$. Our goal is to show that $H$ is at least as dense as $H'$ in $G$. Let $n_i = |V_i|$ and $m_i = |E_i| + \sum_{v \in H_i} w(v)$ for $i \in \{1,2\}$. Then the density of $H_1$ and $H_2$ in the vertex weighted graph $G^*$ is $d_1=m_1/n_1$ and $d_2=m_2/n_2$ respectively. Also, by letting $n_0 = |V_0|$ and $m_0 = |E_0|$, the density of $H_0$ in $G$ is given by $density(H_0)=d_0=m_0/n_0$. Furthermore, observe that by the definition of the vertex weight function in $G^*$, we have 
\begin{align*}
density(H) & = \frac{|E_0|+|E_1|+ |E(V_0,V_1)|}{|V_0|+|V_1|} 
= \frac{|E_0|+|E_1|+ \sum_{v \in V_1}|N(v) \cap V_0|}{|V_0|+|V_1|} =\\
& = \frac{|E_0|+|E_1| + \sum_{v \in V_1}w(v)}{|V_0|+|V_1|}
 = \frac{m_0+m_1}{n_0+n_1}, 
\end{align*}
and similarly, $density(H') = (m_0+m_2)/(n_0+n_2)$. Below we argue that $density(H)$ is at least as large as $density(H')$. 

By standard algebra, we have 
\begin{align*}
density(H) & \geq density(H') & \iff \\
\frac{m_0+m_1}{n_0+n_1} & \geq \frac{m_0+m_2}{n_0+n_2} & \iff\\
m_0n_2 + m_1(n_0 +n_2) & \geq m_0n_1 + m_2(n_0 + n_1)  & \iff\\
m_1n_2 + m_0(n_2-n_1) & \geq m_2n_1 + n_0(m_2-m_1). & 
\end{align*}
Thus, to complete the proof it suffices to prove the following two inequalities: $m_1n_2 \geq m_2n_1$ and $m_0(n_2-n_1) \geq n_0(m_2-m_1)$.

For the first inequality, observe that $d_1 = m_1/n_1 \geq d_2 = m_2/n_2$ as $H_1$ is a densest subgraph in $G^*$;  this directly implies $m_1n_2 \geq m_2n_1$.  For second inequality, by the assumption that $density(H_0) \geq density(H)$, we have:
\begin{align*}
density(H_0) & \geq density(H) & \iff \\
\frac{m_0}{n_0} & \geq \frac{m_0+m_1}{n_0+n_1} & \iff \\
m_0n_1 & \geq m_1n_0 & \iff \\
\frac{m_0}{n_0} & \geq \frac{m_1}{n_1} & \iff \\
density(H_0) & \geq d_1, & 
\end{align*}
Thus, 
$$
\frac{m_0}{n_0} = density(H_0) \geq d_1 = \frac{d_1(n_2-n_1)}{n_2-n_1} \geq \frac{d_2n_2-d_1n_1}{n_2-n_1} = \frac{m_2-m_1}{n_2-n_1},
$$
and so the second inequality also holds. 
\end{proof}

Now, first observe that $G_1$ is computed in $O(mn \log n)$ time
(or $O(n^3)$ time if
$m \log n > n^2$) with Goldberg's algorithm given in \cite{Goldberg:1984:FMD:894477,DBLP:journals/algorithmica/KawaseM18}. Next, note that if some subgraph $G_i$ is big or huge, then the remaining graphs $G_{i+1},\ldots,G_k$ can easily be computed in $O(m+n)$ time. Consider then a small subgraph $G_i$ for some $i \leq k-1$. Then, by construction, all subgraphs $G_1,\ldots,G_i$ are small, and so we have $|V_1| \cdots |V_i|\leq \varepsilon^k = O(1)$. The subgraph $G_{i+1}$ is computed by first computing candidates of two different types. For the first type we need to invoke Goldberg's algorithm on a graph $|V_1| \cdots |V_i| \leq  \varepsilon^k = O(1)$ times, so this requires $O(mn \log n)$ time (or $O(n^3)$  time if
$m \log n > n^2$). For the second type, we need to invoke Goldberg's algorithm 
on a weighted graph, as described
in Lemma~\ref{lem:supergraph} above, $i=O(1)$ times, and so this also requires $O(m n \log n)$ time. In total, we compute each subgraph $G_i$ in $O(mn \log n)$ time, which gives a the same run-time for the entire algorithm since $k=O(1)$.   

\subsection{Approximation-ratio analysis}

Let $G^*_1,\ldots,G^*_k$ be an optimal solution
of \DenseSub{} on instance $G$, with densities $d^*_1 \geq d^*_2 \geq \cdots \geq d^*_k$. We analyze the approximation ratio guaranteed by our algorithm by comparing the density of each subgraph $G_i=(V_i,E_i)$ computed by the algorithm with $d^*_i$. For $G_1$ this is easy. Since $G^*_1$ is a densest subgraph in $G$, and $G_1$ is the graph computed by Goldberg's algorithm, we have:
\begin{lemma}
\label{lem:first}%
$density(G_1) = d^*_1$.
\end{lemma}

For the remaining graphs, our analysis splits into three cases depending on the type of graph previously computed by the algorithm. 

\begin{lemma}
\label{lem:small}%
If $G_i$ is small, for $i < k$, then $density(G_{i+1}) = d^*_{i+1}$.
\end{lemma}

\begin{proof}
The optimal subgraph $G^*_{i+1}=(V^*_{i+1},E^*_{i+1})$ is either a supergraph of some graph in $G_1,\ldots,G_i$, or $V_j \setminus V^*_{i+1} \neq \emptyset$ for each $j \in \{1,\ldots,i\}$. Since the candidates for $G_{i+1}$ considered by our algorithm in case $G_i$ is small cover both these cases, the lemma follows. 
\end{proof}

Note that Lemma~\ref{lem:first} and Lemma~\ref{lem:small} together imply that if all subgraphs computed by the algorithm are small, then $density(G_i)=d^*_i$ for each $i \in \{1,\ldots,k\}$, and our algorithm computes an optimal solution. Furthermore, the first big or huge subgraph it computes also has optimal densities. The next two lemmas deal with the remaining subgraphs that are computed after computing a big or huge subgraph.  

\begin{lemma}
\label{lem:big}
Suppose $G_i$, for $i<k$, is the first big subgraph computed by the algorithm. Then $density(G_j) \geq (1-\frac{1}{\varepsilon}) \cdot d^*_j$ for each $j \in \{i+1,\ldots,k\}$.
\end{lemma}

\begin{proof}
Let $n_i = |V_i|$ and $m_i=|E_i|$. By Lemma~\ref{lem:first} and Lemma~\ref{lem:small} we know that $m_i/n_i=d^*_i$. Furthermore, as $G_i$ is big, we have $n_i > \varepsilon - 1$, or written differently $n_i/(\varepsilon-1) > 1$. Now, as each $G_j$ has $n_i+1$ vertices and at least $m_i$ edges, we have
$$
density(G_j) \geq \frac{m_i}{n_i+1} > \frac{m_i}{n_i+n_i/(\varepsilon-1)} = \frac{\varepsilon-1}{\varepsilon} \cdot \frac{m_i}{n_i} = (1-\frac{1}{\varepsilon}) \cdot d^*_i \geq 
(1-\frac{1}{\varepsilon}) \cdot d^*_j.
$$
\end{proof}

\begin{lemma}
\label{lem:huge}%
Suppose $G_i$, for $i<k$, is the first huge subgraph computed by the algorithm. Then $density(G_j) \geq (1-\frac{1}{\varepsilon}) \cdot d^*_j$ for each $j \in \{i+1,\ldots,k\}$.
\end{lemma}

\begin{proof}
Let $n_i = |V_i|$ and $m_i=|E_i|$. Since $G_i$ is huge we know that $n_i > n-k$, and again by Lemmas~\ref{lem:first} and~\ref{lem:small} we know that $m_i/n_i=d^*_i$. Let $v \in V_i$ be a vertex of minimum degree in $G_i$. Consider the subgraph $G_{i+1}$, constructed from $G_i$ by removing the vertex $v \in V_i$ with minimum degree. Then the degree of $v$ cannot exceed the average degree in $G_i$, and so $deg(v) \leq 2m_i/n_i$. Thus, the density of $G_i$ can be bounded by: 
$$
density(G_{i+1})= \frac{m_i-deg(v)}{n_i-1} \geq \frac{m_i-2m_i/n_i}{n_i-1} = 
\frac{n_i-2}{n_i-1} \cdot  d^*_i > 
\frac{n-k-2}{n} \cdot  d^*_i.
$$
Extending this argument, it can be seen that the density of $G_{i+j}$, for any $j \in \{1,\ldots,k-i\}$, is bounded from below by $\left(\frac{n-k-j-1}{n}\right)^j \cdot  d^*_i$. The lemma then directly follows from Equation~\ref{eqn:sufficientlylarge}. 
\end{proof}

Summarizing, due to Lemmas~\ref{lem:first}, \ref{lem:small}, \ref{lem:big}, and \ref{lem:huge}, we know that $density(G_i) \geq (1-\frac{1}{\varepsilon}) \cdot d^*_i$ for all $i \in\{1,\ldots,k\}$, and so in total we have:
$
\sum^k_{i=1} density(G_i) \geq \sum^k_{i=1} (1-\frac{1}{\varepsilon}) \cdot d^*_i = (1-\frac{1}{\varepsilon}) \cdot OPT. 
$
This completes the proof of Theorem~\ref{thm:distinct}. 

\section{\kDenseSub{} in FPT Time}
\label{sec:FPT}

We next show that \kDenseSub{} is solvable
in $O(2^k m n^3 \log n)$ time, \emph{i.e.} that it is fixed-parameter tractable in $k$. Recall that our goal is to compute $k$ subgraphs $G_1,\ldots,G_k$ of $G=(V,E)$ whose total density is maximal, and our only constraint is that these subgraphs need to be distinct.  

Similarly to Section~\ref{sec:PTAS}, our approach here is to iteratively compute $G_1$, then $G_2$, and so forth, where we start from a densest subgraph $G_1$ of $G$. In what follows, we assume we have already computed the subgraphs $G_1=(V_1,E_1),\ldots,G_\ell=(V_\ell,E_\ell)$, for $\ell \in \{1,\ldots,k-1\}$, and our goal is to compute a densest subgraph $G_{\ell+1}=(V_{\ell+1},E_{\ell+1})$ among all subgraphs in $G$ distinct from $G_1,\ldots,G_\ell$. Let $V^*=\bigcup^\ell_{i=1} V_i$. We consider the following two cases: 
\begin{enumerate}
\item There is some vertex $v \in V_{\ell+1}$ that is not in $V^*$, \emph{i.e.} $V_{\ell+1} \nsubseteq V^*$. 
\item $V_{\ell+1}$ is contained completely in $V^*$, \emph{i.e.} $V_{\ell+1} \subseteq V^*$. 
\end{enumerate}
We compute a densest subgraph in each one of these cases, and then take the densest of the two to be $G_{\ell+1}$. 

\subsection{First case}

The first case where $V_{\ell+1} \nsubseteq V^*$ is easy: we iterate through all vertices $v \in V \setminus V^*$ and compute a densest subgraph 
of $G$ that includes $v$, and then take the densest of all these subgraphs (each of them being distinct from $G_1,\ldots,G_\ell$).

\begin{lemma}
\label{lem:includes_v}%
Let $v \in V$. A densest subgraph of $G$ that includes $v$ can be computed in $O(m n \log n)$ time.
\end{lemma}
\begin{proof}
Let $w_v:V \to \mathbb{N}$ be the weight function defined by $w_v(v)=n^2$, and $w_v(u)=1$ for all vertices $u \neq v$. Then any subgraph of $G$ that does not include $v$ has weighted density less than $n$, and any subgraph that includes $v$ has weight density at least $n$. It follows that computing a densest subgraph of $G$ that includes $v$ can be done by a single application of Goldberg's algorithm in 
$O(m n \log n)$ time on $G$ weighted by $w_v$. 
\end{proof}

\begin{lemma}
\label{lem:case1}%
If $V_{\ell+1} \nsubseteq V^*$ then $G_{\ell+1}$ can be computed in $O(mn^2 \log n)$ time.
\end{lemma}
\begin{proof}
Iterate on all $O(n)$ vertices $v \in V \setminus V^*$, and run the algorithm in Lemma~\ref{lem:includes_v} for each such vertex $v$. In total, by Lemma~\ref{lem:includes_v} this takes $O(n) \cdot O(mn \log n) = O(mn^2 \log n)$ time.
\end{proof}

\subsection{Second case}

The second case where $V_{\ell+1} \subseteq V^*$ requires more details. We say that a non-empty subset $\mathcal{C} \subseteq \{V_1,\ldots,V_\ell\}$ \emph{covers} $V_{\ell+1}$ if $V_{\ell+1} \subseteq V_\mathcal{C}=\bigcup_{V_i \in \mathcal{C}} V_i$, and it is a \emph{minimal cover} if $V_{\ell+1} \nsubseteq V_\mathcal{C'}$ for any proper subset $\mathcal{C'} \subset \mathcal{C}$. Our approach is to compute for each non-empty subset $\mathcal{C} \subseteq \{V_1,\ldots,V_\ell\}$, a densest subgraph of $G$ for which $\mathcal{C}$ is a minimal cover.

\begin{lemma}
\label{lem:vinvout}
Let $\mathcal{C}\subseteq\{V_1,\ldots,V_\ell\}$, and suppose that $\mathcal{C}$ is a minimal cover of~$V_{\ell+1}$. If $V_{\ell+1} \neq V_\mathcal{C}$, then there are two vertices $v_{in}, v_{out} \in V_\mathcal{C}$ such that $v_{in} \in V_{\ell+1}$ and $v_{out} \notin V_{\ell+1}$, and there is no subset $V_i \in \mathcal{C}$ with $v_{in} \in V_i$ and $v_{out} \notin V_i$.
\end{lemma}
\begin{proof}
Suppose that $V_{\ell+1} \neq V_\mathcal{C}$, and so $V_{\ell+1} \subset V_\mathcal{C}$. It follows that there exists a vertex $v_{out} \in V_\mathcal{C} \setminus V_{\ell+1}$. Consider the subset $\mathcal{C'} \subset \mathcal{C}$ which includes all vertex subsets in $\mathcal{C}$ that do not include the vertex $v_{out}$, \emph{i.e.} $\mathcal{C'}=\{V_i \in \mathcal{C} : v_{out} \notin V_i\}$. Note that  $\mathcal{C'}$ is indeed a proper subset of     $\mathcal{C}$, as $v_{out}$ belongs to some graph in $\mathcal{C}$. If $\mathcal{C'} = \emptyset$, then $v_{out}$ belongs to every subset $V_i \in \mathcal{C}$, and the lemma holds. If $\mathcal{C'} \neq \emptyset$, there must be some vertex $v_{in} \in V_{\ell+1} \setminus V_{\mathcal{C'}}$ by the minimality of $\mathcal{C}$, since otherwise $\mathcal{C'}$ would cover $V_{\ell+1}$. 
\end{proof}

\begin{lemma}
\label{lem:case2}%
If $V_{\ell+1} \subseteq V^*$ then $G_\ell$ can be computed in $O(2^k m n^3 \log n)$ time.
\end{lemma}
\begin{proof}
We iterate over all possible $2^\ell-1$ non-empty subsets  $\mathcal{C} \subseteq \{V_1,\ldots,V_\ell\}$. For each subset $\mathcal{C}$, we iterate over all $O(n^2)$ vertices $v_{in}, v_{out} \in V_\mathcal{C}$ and compute a densest subgraph in $G[V_\mathcal{C}\setminus\{v_{out}\}]$ that includes $v_{in}$ (using the algorithm in Lemma~\ref{lem:includes_v}). This requires $O(m n^3 \log n)$ time in total. Out of all subgraphs computed this way, along with all subgraphs of the form $G[V_{\mathcal{C}}]$, we choose the densest subgraph which is distinct from $\{G_1,\ldots,G_\ell\}$.  As $G_{\ell +1}$ is a densest subgraph in $G[V_\mathcal{C}\setminus\{v_{out}\}]$ that includes $v_{in}$, for the minimal cover $\mathcal{C}$ of $V_{\ell+1}$ and some $v_{in},v_{out} \in V_\mathcal{C}$ (according to Lemma~\ref{lem:vinvout}), this algorithm is indeed guaranteed to find a subgraph of $G$ with density at least $density(G_{\ell+1})$. 
\end{proof}

\subsection{Summary}

Thus, taking the densest of the subgraph given by Lemma~\ref{lem:case1} and the subgraph given by Lemma~\ref{lem:case2} gives us a densest subgraph in $G$ which is distinct from $\{G_1,\ldots,G_\ell\}$ in $O(2^k m n^3 \log n)$ time. In this way, we can compute $k$ densest distinct subgraphs of $G$ in $O(2^k k m n^3 \log n)$ time, completing the proof of Theorem~\ref{thm:distinctFPT}.

\section{Conclusion}
\label{sec:conclusion}

This paper studies a natural variant for computing $k$ densest subgraphs of a given graph, a central problem in graph data mining. 
We show that the problem is fixed-parameter tractable with respect to $k$, and admits a PTAS for $k=O(1)$. 

From a theoretical perspective, the most interesting problem that is left open by our paper is whether 
\kDenseSub{} is NP-hard for unbounded $k$. However, we feel that for most practical settings, the number~$k$ of solution subgraphs should be significantly smaller than the size~$n$ of the network. Thus, we feel that examining the problem 
on specific social network models might be more interesting from a practical point of view.  
Finally, we have considered unweighted graphs, 
a natural direction is whether it is possible
to extend the results to edge-weighted graphs.

\section*{Acknowledgements}

We thank an anonymous reviewer for pointing out an error in an algorithm included in a previous version of the paper.

\bibliography{biblio1}

\end{document}